\documentclass[journal, draftcls, 12 pt, onecolumn, twoside]{IEEEtran} 
\usepackage{cite}

\ifCLASSINFOpdf
  \usepackage[pdftex]{graphicx}
\else
\fi
\usepackage[cmex10]{amsmath}
\usepackage{array}

\usepackage{mdwmath}
\usepackage{mdwtab}

\usepackage{eqparbox}

\usepackage[tight,footnotesize]{subfigure}

\usepackage[caption=false,font=footnotesize]{subfig}
\usepackage{fixltx2e}

\usepackage{stfloats}
\usepackage{url}

\usepackage[dvipsnames]{color}
\usepackage{amsfonts} 
\usepackage{epstopdf} 
\usepackage{amssymb} 
\usepackage{mathrsfs} 
\usepackage{bbm} 
\usepackage{dsfont} 
\usepackage{acronym} 
\usepackage{enumerate} 
\usepackage{algcompatible} 
\usepackage{algorithm}
\usepackage{algpseudocode}
\usepackage{caption} 

\usepackage[latin1]{inputenc}
\usepackage{graphicx,psfrag,cite,subfigure,stfloats}
\usepackage{soul}

\DeclareSymbolFont{bbold}{U}{bbold}{m}{n}
\DeclareSymbolFontAlphabet{\mathbbold}{bbold}

\acrodef{RV}{random variable}

\acrodef{i.i.d.}{independent, identically distributed}

\acrodef{p.d.f.}{probability density function}

\acrodef{p.m.f.}{probability mass function}

\acrodef{c.d.f.}{cumulative distribution function}

\acrodef{FT}{Fourier Transform}

\acrodef{ch.f.}{characteristic function}

\acrodef{AWGN}{additive white Gaussian noise}

\acrodef{SNR}{signal-to-noise ratio}

\acrodef{LRT}{likelihood ratio test}

\acrodef{LSCR}{Leave-out Sign-dominant Correlation Regions}

\acrodef{GLRT}{generalized likelihood ratio test}

\acrodef{LOS}{line-of-sight}

\acrodef{NLOS}{non-line-of-sight}

\acrodef{GF}{Greedy Flooding}

\acrodef{GDOP}{geometric dilution of precision}

\acrodef{GPS}{Global Positioning System}

\acrodef{FIM}{Fisher information matrix}

\acrodef{PEB}{position error bound}

\acrodef{WSN}{Wireless Sensor Network}

\acrodef{MAC}{medium access control}

\acrodef{RSS}{received signal strength}

\acrodef{RTT}{round-trip time}

\acrodef{MF}{Modified Flooding}

\acrodef{PF}{Plain Flooding}

\acrodef{ED}{energy detector}

\acrodef{ML}{maximum likelihood}

\acrodef{NL}{nonlinear}

\acrodef{MSE}{mean square error}

\acrodef{RMSE}{root mean square error}

\acrodef{ppm}{part-per-million}

\acrodef{ACK}{acknowledge}

\acrodef{UWB}{ultrawide bandwidth}

\acrodef{TNR}{threshold-to-noise ratio}

\acrodef{NLOS}{non line-of-sight}

\acrodef{LOS}{line-of-sight}

\acrodef{LS}{least squares}

\acrodef{IR-UWB}{impulse radio UWB}

\acrodef{FCC}{Federal Communications Commission}

\acrodef{TH}{time-hopping}

\acrodef{p.m.}{probability measure}

\acrodef{PPM}{pulse position modulation}

\acrodef{PAM}{pulse amplitude modulation}

\acrodef{MUI}{multi-user interference}

\acrodef{PDP}{power delay profile}

\acrodef{BPZF}{band-pass zonal filter}

\acrodef{SIR}{signal-to-interference ratio}

\acrodef{RFID}{radiofrequency identification}

\acrodef{WPAN}{wireless personal area networks}

\acrodef{WWLB}{Weiss-Weinstein lower bound}

\acrodef{DP}{direct path}

\acrodef{MMSE}{minimum-mean-square-error}

\acrodef{SBS}{serial backward search}

\acrodef{SF}{Smart Flooding}

\acrodef{NBI}{narrowband interference}

\acrodef{WBI}{wideband interference}

\acrodef{INR}{interference-to-noise ratio}

\acrodef{CIR}{channel impulse response}

\acrodef{LRT}{likelihood ratio test}

\acrodef{RADAR}{RADAR}

\acrodef{MUR}{Multistatic RADAR}

\acrodef{e.m.}{electromagnetic}

\acrodef{CW}{continuous wave}

\acrodef{RF}{radiofrequency}

\acrodef{FCC}{Federal Communications Commission}

\acrodef{EIRP}{effective radiated isotropic power}

\acrodef{RCS}{radar cross section}

\acrodef{BAV}{balanced antipodal Vivaldi}

\acrodef{PRake}{partial Rake}

\acrodef{RTLS}{real time locating systems}

\acrodef{Hi-RADIAL}{High-accuracy RAdio Detection, Identification,
And Localization}

\acrodef{CRB}{Cram\'{e}r-Rao bound}

\acrodef{ZZB}{Ziv-Zakai bound}

\acrodef{TOA}{time-of-arrival}

\acrodef{TOF}{time-of-flight}

\acrodef{WSN}{wireless sensor network}

\acrodef{MAC}{medium access control}

\acrodef{RSS}{received signal strength}

\acrodef{TAS}{tagged and aggregated sums}

\acrodef{TDOA}{time difference-of-arrival}

\acrodef{RF}{radiofrequency}

\acrodef{RTT}{round-trip time}

\acrodef{AOA}{angle-of-arrival}

\acrodef{ED}{energy detector}

\acrodef{ML}{maximum likelihood}

\acrodef{MUR}{Multistatic radar}

\acrodef{HDSA}{high-definition situation-aware}

\acrodef{RRC}{root raised cosine}

\acrodef{OFDM}{orthogonal frequency division multiplexing}

\acrodef{IF}{intermediate frequency}

\acrodef{PHY}{physical layer}

\acrodef{S-V}{Saleh-Valenzuela}

\acrodef{UHF}{ultra-high frequency}

\acrodef{PR}{pseudo-random}

\acrodef{PPP}{Poisson Point Process}

\acrodef{SoC}{System on Chip}

\acrodef{SoP}{System on Package}

\acrodef{SPMF}{Single-Path Matched Filter}

\acrodef{SPS}{sign perturbed sums}

\acrodef{IMF}{Ideal Matched Filter}

\acrodef{SCR}{signal-to-clutter ratio}

\acrodef{BEP}{bit error probability}

\acrodef{BER}{bit error rate}

\newtheorem{Property}{Property}

\newtheorem{remark}{Remark}
\newtheorem{definition}{Definition}
\newtheorem{lemma}{Lemma}
\newtheorem{theorem}{Theorem}

\begin{document}

\newcommand{\ud}{\mathrm{d}}  
\newcommand{\blue}[1] {{\color{blue}{#1}}}
\newcommand{\red}[1] {{\color{red}{#1}}}
\definecolor{viol}{rgb}{0.7,0,1}
\newcommand{\viol}[1] {{\color{viol}{#1}}}

%
\title{Efficient Distributed Non-Asymptotic Confidence Regions Computation over Wireless Sensor Networks}
%
%
%

\author{Vincenzo~Zambianchi,~\IEEEmembership{Student Member,~IEEE,}
        Michel Kieffer,~\IEEEmembership{Senior Member,~IEEE,}
 		Gianni Pasolini,~\IEEEmembership{Senior Member,~IEEE,}
        Francesca Bassi,~\IEEEmembership{Member,~IEEE,}
   		and Davide~Dardari,~\IEEEmembership{Senior Member,~IEEE}
\thanks{V. Zambianchi, D. Dardari and G. Pasolini are with the Department
of Electrical, Electronics and Information Engineering ``G. Marconi'' (DEI), University of Bologna, via Venezia 52, Cesena (FC) 47521,
Italy. E-mail: \{vincenzo.zambianchi, davide.dardari, gianni.pasolini\}@unibo.it.
M. Kieffer is with CNRS, Supelec, Univ. Paris-Sud, France. E-mail: Michel.Kieffer@lss.supelec.fr.
F. Bassi is with ESME Sudria, Supelec, France. E-mail: francesca.bassi@lss.supelec.fr.
This work has been supported by the EU funded Newcom\# Project.}}

%
%

\markboth
    {submitted to IEEE Transactions on Wireless Communications}
    {Zambianchi et al.: Efficient Distributed Non-Asymptotic Confidence Regions Computation }
%



\maketitle

\begin{abstract}
This paper considers the distributed computation of confidence regions tethered to multidimensional parameter estimation under linear measurement models. In particular, the considered confidence regions are non-asymptotic, this meaning that the number of required measurements is finite. Distributed solutions for the computation of non-asymptotic confidence regions are proposed, suited to wireless sensor networks scenarios. Their performances are compared in terms of required traffic load, both analytically and numerically. The evidence emerging from the conducted investigations is that the best solution for information exchange depends on whether the network topology is structured or unstructured. The effect on the computation of confidence regions of information diffusion truncation is also examined. In particular, it is proven that consistent confidence regions can be computed even when an incomplete set of measurements is available.
\end{abstract}

%

%
\IEEEpeerreviewmaketitle

\section{Introduction}
%
%
%
%

\IEEEPARstart {A}{ }\ac{WSN} consists of energy-limited sensing devices deployed to collaborate in performing  a common task. Examples may be the monitoring of an environmental parameter (\emph{e.g.} temperature or pressure \cite{Akyildiz:2002, Mainwaring:2002,VerDarMazCon:B08}), the detection of a binary event \cite{QueDarWin:07}, the estimation of a spatial field \cite{MatFabAntDar:J11}, the estimation of the coordinates of a signal source \cite{MaoFidAnd:07}, etc.

Depending on the specific task requirements (fault tolerance, privacy issues, energy constraints), either a centralized or a distributed approach can be adopted: In the former a central unit is needed, that collects all the information and completes the objective task, whereas in the latter all nodes accomplish the objective task on the basis of the information previously exchanged among them. In the centralized scenario the adoption of efficient routing schemes is of capital importance. Fundamental contributions in this sense are the energy-efficient adaptive clustering proposed in \cite{HeiChaBal:02} and the routing protocols in \cite{MadLal:06, YanPraKri:06, MinLeuMal:07, KwoShr:12}, aimed at extending network lifetime.

One of the most studied topic in the \ac{WSN} literature is the estimation of physical parameters. The literature is mostly focused on the development of some specific estimation techniques, both for the centralized and distributed approaches. Classical \ac{ML} or \ac{LS} estimators \cite{Kay:2013} work under the hypothesis of having all the required observations available at one central unit. The scarce robustness to central unit failures and poor network scalability have brought to consideration of distributed approaches. For instance, \cite{schizas09,mateos09} address recursive weighted \ac{LS} estimation, alongside a consensus-based algorithm that allows to incorporate information from neighbor nodes in the local estimate. A similar approach is taken within the Bayesian framework in \cite{OlfFaxMur:07,Olf:07b,Olf:09}, where consensus-based distributed Kalman filtering is proposed.

The distributed computation of confidence regions has been less considered:
In some applications, however, (\emph{e.g.,} in source localization) the derivation of the confidence region is as important as the determination of the estimate. Classical Cramer-Rao-like bounds have been proposed to this purpose in \cite{YanSch:05,SheHu:05,PatAshKyp:05,JouDarWin:08}. Confidence regions can also be derived as a by-product of the application of Kalman filtering \cite{Olf:07b,Olf:09}. However, strong assumptions on measurement noise (typically Gaussian) are necessary and a good characterization of confidence regions is only possible for a large number of measurements (asymptotic regime).

If we restrict the attention to the centralized setup, the derivation of confidence regions in the non-asymptotic regime has been proved to be possible using, for example, the results in \cite{CamWey:05,CamKoWey:09, CsaCamWey:12, KieWal:14}.
Specifically, the methods proposed in \cite{CamWey:05, CamKoWey:09} allow the central unit to derive a confidence region and a lower bound on the probability that the true value of the estimated parameter falls within it, whereas the \emph{exact} probability
can be obtained using the \ac{SPS} algorithm in \cite{CsaCamWey:12}.
In \cite{KieWal:14}, an efficient centralized computation of confidence regions is obtained using
interval analysis techniques. Differently from Cramer-Rao-like bounds, these methods
do not require precise statistical knowledge of the noise, and work under very mild assumptions on its distribution.


\subsection{Main Contributions}
Some preliminary results on the derivation of exact non-asymptotic confidence regions, in a distributed scenario, appeared in \cite{ZamKieBasPasDar:14}.
%
To ensure that the confidence region computed by each node is similar in shape to the one that would be evaluated in a centralized setup, nodes have to share their local information with one another. The way of diffusing information drastically impacts on the amount of data exchanged. For this reason, several information diffusion strategies are analyzed and compared in the following. 
%
A novel information diffusion strategy, named \ac{TAS}, is presented. It exploits the peculiarities of the \ac{SPS} algorithm, allowing a reduction of the amount of information to be exchanged among nodes. Its performance is compared to that of established information diffusion strategies, such as flooding \cite{Akyildiz:2002,Heinzelman:1999} and consensus algorithms \cite{OlfFaxMur:07}, in terms of generated traffic load as well as confidence region volume/traffic trade-off. Performance predictions and simulation results are provided for various topologies. The introduction of the \ac{TAS} algorithm is one of the novelties of this work.

Constraints on traffic load may lead to information diffusion truncation: Certain nodes might hence compute a confidence region with partial data. However, we prove that
consistent non-asymptotic confidence regions can be computed, even starting from an incomplete set of measurements. This constitutes a second theoretical novel contribution to be found in this paper.

The remainder is organized as follows. Section~\ref{sec:Recall} formulates the confidence region computation problem and recalls the \ac{SPS} algorithm. Section~\ref{sec:InfoDiffAlg} presents information diffusion strategies. The computation of non-asymptotic confidence regions, from an incomplete set of measurements, is analyzed in Section~\ref{sec:Trunc}. Information diffusion techniques are compared on various network topologies in Sections~\ref{sec:ThAnalysis} and~\ref{sec:SimRes}. Conclusions are drawn in Section~\ref{sec:Concl}.

\subsection{Notation}
In this paper, \acp{RV} are indicated with capital roman or greek letters. Their realizations are denoted by the corresponding lowercase letters. Vectors are denoted by bold letters, being lowercase or uppercase according to their random or deterministic nature, while matrices are indicated with bold capital letters.

\section{Problem Formulation}\label{sec:Recall}
This section recalls the centralized \ac{SPS} algorithm \cite{CsaCamWey:12} for the computation of non-asymptotic confidence regions. Consider some spatial field described by the parametric model
\begin{align}
y_{\text{m}}\left(\mathbf{x},\mathbf{p}\right)=\boldsymbol{\varphi}^{T}\left(\mathbf{x}\right)\mathbf{p},\label{eq:Model}
\end{align}
where $\mathbf{x}\in\mathbb{R}^{n_x}$ represents some vector of experimental conditions (time, location, \ldots) under which the
field is observed, $\boldsymbol{\varphi}\left(\mathbf{x}\right)$
is some regressor function, and $\mathbf{p}$ is
the vector of unknown parameters, belonging to the parameter space $\mathbb{P}\subset \mathbb{R}^{n_p}$. For further discussions on the adopted linear model, one may refer to \cite{NorChiVit:08} and references therein.

Measurements are taken by a network of $N$ sensor nodes, spread at
random locations $\mathbf{x}_{i} \in \mathbb{R}^{n_x}$, $i=1,\dots,N$. Each sensor collects its scalar measurement $y_{i}$ according to the local measurement model
\begin{align}
Y_{i} & =  y_{\text{m}}\left(\mathbf{x}_{i},\mathring{\mathbf{p}}\right)+W_{i} 
 =\boldsymbol{\varphi}^{T}_i\mathring{\mathbf{p}}+W_{i},
\end{align}
where $\boldsymbol{\varphi}^{T}_i= \boldsymbol{\varphi}^{T}\left(\mathbf{x}_{i}\right)$ is the
regressor vector at $\mathbf{x}_{i}$, assumed to be known, $\mathring{\mathbf{p}}$
is the true value of the parameter vector and $W_{i}$ is a random
variable representing the measurement noise. The only assumption on $W_{i}$s is that they are independent from node to node with a distribution, whichever its shape, symmetric with respect to zero.

The aim of this paper is to investigate the distributed derivation of exact
non-asymptotic confidence regions, keeping as low as possible the amount of data that has to be exchanged among sensors.
As starting point, we recall the centralized \ac{SPS} algorithm \cite{CsaCamWey:12} that assumes all measurements and regressors to be known to a central processing unit and returns the exact confidence region around
the least squares estimate $\widehat{\mathbf{p}}$ of $\mathring{\mathbf{p}}$,
obtained as the solution of the normal equations $\sum_{i=1}^{N}\boldsymbol\varphi_{i}\left(y_{i}-\boldsymbol\varphi_{i}^{T}\mathbf{p}\right)=\mathbf{0}$.   Specifically, \cite{CsaCamWey:12} introduces 
the unperturbed sum
\begin{align}
\mathbf{S}_{0}(\mathbf{p})=\sum_{i=1}^{N}\boldsymbol\varphi_{i}\left(Y_{i}-\boldsymbol\varphi_{i}^{T}\mathbf{p}\right)\label{s0}
\end{align}
and the $m-1$ sign-perturbed sums, for some $m$, with $2\leq m \leq N$,
\begin{align}
\mathbf{S}_{j}(\mathbf{p})=\sum_{i=1}^{N}A_{j,i}\boldsymbol\varphi_{i}\left(Y_{i}-\boldsymbol\varphi_{i}^{T}\mathbf{p}\right), \; j=1,\ldots,m-1 \label{sj}
\end{align}
where $A_{j,i}\in\{\pm1\}$ are independent random signs\footnote{A random sign is a symmetric $\pm 1$ valued random variable taking both values with the same probability $1/2$.
}. Introducing
\begin{align}
Z_{j}(\mathbf{p})=||\mathbf{S}_{j}(\mathbf{p})||_{2}^{2},\;\,\,\,\,\,\,\,\, j=0,\ldots,m-1,\label{zj}
\end{align}
one may define the set
\begin{align}
\mathbf{\Sigma}_{q} & =\left\{ \mathbf{p}\in\mathbb{P}|Z_{0}(\mathbf{p})\text{ is not among the \ensuremath{q} largest }Z_{j}(\mathbf{p})\right\} \nonumber \\
 &= \left\{ \mathbf{p}\in\mathbb{P}\left|\quad\sum_{j=1}^{m-1}\mathbb{I}(Z_{j}(\mathbf{p})-Z_{0}(\mathbf{p}))\geq q\right.\right\} ,\label{eq:Test}
\end{align} 
where 
$\mathbb{I}(\cdot)$ is the indicator function on positive reals. In \cite{CsaCamWey:12}, it was proven that
\begin{align}
\text{Prob}(\mathring{\mathbf{p}}\in\mathbf{\Sigma}_{q})=1-\frac{q}{m}.\label{eq:prob}
\end{align}
As a consequence $\mathbf{\Sigma}_{q}$ is a non-asymptotic confidence region with confidence level $1-q/m$.  

In the following, the distributed computation of $\mathbf{\Sigma}_{q}$ will be addressed considering different information diffusion strategies.

\section{Information Diffusion Algorithms} \label{sec:InfoDiffAlg}
This section describes concurrent procedures for information diffusion adapted to \ac{SPS}. The purpose is to let each node capable of collecting the largest amount of measurements $y_i$ and regressors $\boldsymbol\varphi_{i}$ possibly with the lowest amount of data exchanged in the network.
For each presented algorithm, the evolution of the amount of information available at a node $k$ is described by a table $\mathbf{R}^{(k)}$. The construction of $\mathbf{R}^{(k)}$ and the transmission of information depend on the considered procedure.
%

\subsection{\acf{PF} Algorithm}
When adopting this simple information diffusion strategy \cite{Heinzelman:1999,Akyildiz:2002,ZamKieBasPasDar:14} each node broadcasts in turn its own measurement and regressor, \emph{i.e.} $\mathbf{D}^{(k)}=\left[ \boldsymbol\varphi_k^T,y_k \right]$, as well as those received from other nodes in previous rounds. This strategy is the most trivial one but does not result to be particularly efficient. On lossless networks, it is outperformed by the following one, and is therefore no more considered in the remainder of this work.

\begin{table}
\begin{center}
$\mathbf{R}^{(1)}=\;\;$\begin{tabular}{|c||c|c|c|c|c|c|c|}
\hline
$\mathbf{D}_1^{(1)}$ & 1 & 0 & 0 & 0 & 0 & 0 & 0 \tabularnewline
\hline
$\mathbf{D}_2^{(1)}$ & 0 & 0 & 1 & 0 & 0 & 0 & 0 \tabularnewline
\hline
$\mathbf{D}_3^{(1)}$ & 0 & 0 & 0 & 0 & 0 & 1 & 0 \tabularnewline
\hline
$\mathbf{D}_4^{(1)}$ & 0 & 1 & 0 & 0 & 0 & 0 & 0 \tabularnewline
\hline
$\mathbf{D}_5^{(1)}$ & 0 & 0 & 0 & 1 & 0 & 0 & 0 \tabularnewline
\hline
$\mathbf{D}_6^{(1)}$ & 0 & 0 & 0 & 0 & 0 & 0 & 1 \tabularnewline
\hline
$\mathbf{D}^{(1)}_7$ & 0 & 0 & 0 & 0 & 1 & 0 & 0 \tabularnewline
\hline
\end{tabular}
\end{center}
\caption{Table $\mathbf{R}^{(1)}$ of available information at node $k=1$ when \ac{MF} is used for information diffusion.}
\label{TableMF}
\end{table}
 
\subsection{\acf{MF} Algorithm} 
The main difference between the \ac{MF} and the \ac{PF} is that in the former, an information already transmitted by a node is never transmitted again by the same node. This kind of behavior is certainly efficient in terms of amount of data to be transmitted on lossless links.

The \ac{MF} algorithm generates at runtime a table of contents available at nodes. An example for this is depicted in Table~\ref{TableMF}. This table gathers the information collected at node $k=1$ in a network composed of 7 nodes. Each row $r$ in $\mathbf{R}^{(k)}$ contains an available information $\mathbf{D}_r^{(k)}$ and its related tag, indicating the originating node.
When performing the \ac{MF} algorithm, the generic node $k$ initially fills the first line of $\mathbf{R}^{(k)}$ with its own local information, \emph{i.e.},
$\mathbf{D}_{1}^{(k)}=\left [ \boldsymbol\varphi_{k}^T,y_{k}\right ]$ and the corresponding tag\footnote{We denote the tag matrix by $\mathbf{T}^{(k)}$ and its $r$-th row by $\mathbf{t}_r^{(k)}$.} $\mathbf{t}_1^{(k)}$ having a single 1 at the $k$-th entry. It then broadcasts $\left\{ \mathbf{D}_{1}^{(k)}, \mathbf{t}_1^{(k)} \right\}$ and marks the line as already transmitted.
As next step, it collects the data coming from neighbors and inserts in the table this new acquired information, thus creating a set of rows corresponding to its set of neighbor nodes, here denoted as $\mathcal{N}_k$. Then it forwards a new data packet containing the data of all lines in $\mathbf{R}^{(k)}$ which were not marked as already transmitted. This means that the second message that node $k$ transmits contains $ \left\{ \mathbf{D}_{j}^{(k)}, \mathbf{t}_j^{(k)}\right\}_{j\in \mathcal{N}_k}$. All rows whose data have been transmitted are then marked.  
The iteration of the procedure yields, at the next transmission step, a message to be transmitted containing only information never previously transmitted. This process terminates when each node in the network has collected
the information from all nodes. Section~\ref{sec:Trunc} analyses the case when all data cannot be gathered at all nodes due, e.g., to delay/traffic constraints.

Afterwards, each node is able to compute the perturbed and unperturbed sums
in \eqref{s0} and \eqref{sj} for any $\mathbf{p}$, and hence derive the confidence region.
During the first iteration, each node has to transmit a packet containing
\begin{align}\label{dMF}
d_{MF}=n_p+1
\end{align}
real values. The dimension of successive data packets is an integer multiple of this value, possibly zero.
\begin{remark}
If all nodes agree on their random generators seed, the computed confidence regions are the same at all nodes without any need for transmission of $A_{j,i}$. In case this agreement is lacking, still transmission of $A_{j,i}$ can be avoided, but the shape of confidence regions computed at different nodes may differ.
\end{remark}


\subsection{Tagged and aggregated sums (\ac{TAS}) Algorithm} \label{Subsec:TAS}

Before coming to the detailed description of the \ac{TAS} algorithm, a preliminary consideration is needed. Expanding a realization of \eqref{s0} and \eqref{sj} one gets,
\begin{align}
&\mathbf{s}_{0}(\mathbf{p})=\sum_{k=1}^{N}\boldsymbol\varphi_{k}y_{k}- \left(\sum_{k=1}^{N} \boldsymbol\varphi_{k} \boldsymbol\varphi_{k}^{T}\right)\mathbf{p} \label{s0exp} \\ 
&\mathbf{s}_{j}(\mathbf{p})=\sum_{k=1}^{N} a_{j,k}\boldsymbol\varphi_{k} y_{k}- \left(\sum_{k=1}^{N} a_{j,k}\boldsymbol\varphi_{k} \boldsymbol\varphi_{k}^{T}\right)\mathbf{p}, \; j=1,\ldots, m-1 \label{sjexp}
\end{align}
The  evaluation of \eqref{s0exp} and \eqref{sjexp} for any value of $\mathbf{p}$ $\in \mathbb{P}$ does not necessarily require the availability of each individual term but rather of 
\begin{align}\label{dataTAS}
\left \{ \sum_{k=1}^{N} \boldsymbol\varphi_{k}y_{k}, \sum_{k=1}^{N} \boldsymbol\varphi_{k}\boldsymbol\varphi_{k}^{T} ,\left \{ \sum_{k=1}^{N} a_{j,k}\boldsymbol\varphi_{k}y_{k} \right \}_{j=1,\ldots,m-1}, \left \{ \sum_{k=1}^{N} a_{j,k}\boldsymbol\varphi_{k}\boldsymbol\varphi_{k}^{T} \right \}_{j=1,\ldots, m-1} \right \}.
\end{align} 
Therefore, at each information diffusion step, the available information can be composed into an aggregated sum, reducing the traffic load. This is the peculiarity of the \ac{SPS} algorithm that can be exploited by both the \ac{TAS} and the consensus algorithms. The main difficulty lies in avoiding the same term to appear more than once in each sum, independently of network topology. This consideration led to the formulation of the \ac{TAS} algorithm whose details follow.

The \ac{TAS} algorithm consists of six phases, namely, i) initialization, ii) reception, iii) distillation, iv) aggregation, v) transmission, and vi) wrap-up, introduced hereafter.

\vskip 1.5mm i) \textit{Initialization phase}. During the initialization phase each node $k \in\{1,...,N\}$ creates and transmits a data packet which consists of the first row of its table $\mathbf{R}^{(k)}$. This first row is composed of:

\begin{itemize}
 \item  a \textit{data set} $\mathbf{D}^{(k)}_1 = \left\{ \boldsymbol\varphi_{k} y_{k} , \left \{ \boldsymbol\varphi_{k} \boldsymbol\varphi_{k}^{T} \right \} , \left \{ a_{j,k} \boldsymbol\varphi_{k} y_{k} \right \}_{j=1, \ldots,m-1}, \left \{ a_{j,k} \boldsymbol\varphi_{k} \boldsymbol\varphi_{k}^{T} \right \}_{j=1, \ldots,m-1} \right\}$, corresponding to the local quantities related to node $k$. This set consists of 
\begin{align}\label{dTAS}
d_{\text{TAS}}=m\left(n_p+n_p\frac{n_p+1}{2}\right)
\end{align}
real values. This computed dimension takes into account the symmetry of $\boldsymbol\varphi_{k}\boldsymbol\varphi_{k}^{T}$. The dimension of data sets obtained as sums of initial data sets does not vary and stays equal to $d_{\text{TAS}}$.
 \item a \textit{tag vector} $\mathbf{t}^{(k)}_1$, that is an all-zero vector except for the $k$-th entry where a 1 is located.
\end{itemize}
After initialization, the reception, distillation, aggregation, and transmission phases are sequentially repeated until a termination condition is met. Within each cycle, new rows \{$\mathbf{D}^{(k)}_r$, $\mathbf{t}^{(k)}_r$\} are possibly added to $\mathbf{R}^{(k)}$, with $r>1$ representing the row number.

For $r>1$, the $r$-th \textit{data set} $\mathbf{D}^{(k)}_r$ can either contain the local quantities related to another node or the sum of quantities related to several nodes, as specified in $\mathbf{t}^{(k)}_r$.

%

\vskip 1.5mm ii) \textit{Reception phase}. During this phase each node collects \textit{messages} transmitted by its neighbors. The message $\mathbf{m}^{(n)}$ coming from node $n$, with $n \neq k$, consists of a \textit{data set} and a \textit{tag vector}.

\vskip 1.5mm iii) \textit{Distillation phase}. At the end of the reception phase, at each node, the received \textit{tag vectors} and those already stored in $\mathbf{R}^{(k)}$ are compared to detect whether the received data contains new information. More precisely, the received \textit{tag vector} of each incoming \textit{message} is compared to all the already available \textit{tag vectors} $\mathbf{t}^{(k)}_r$ contained in $\mathbf{R}^{(k)}$. If a received \textit{message} does not contain any new contribution it is discarded, otherwise a new row is added to $\mathbf{R}^{(k)}$, containing the new information contribution, that is, the received \textit{message} (\textit{data set}+\textit{tag vector}) duly polished of already available information.

\emph{Example 1}: if node $k$ receives a \textit{message} containing the
sum of quantities originating from nodes 1, 2, 7, 8, 11 and if it has already
rows in $\mathbf{R}^{(k)}$ containing the information relative to
node 1 and to the sum of local quantities of nodes 2 and 7, it can successfully detect
the sum of quantities related to nodes 8 and 11 and insert them in the table. Only this
distilled information, composed of a new \textit{data set} and its corresponding \textit{tag vector} (having the 8th and 11th bits set
to 1) is added to $\mathbf{R}^{(k)}$.

\begin{table}
\begin{center}
\begin{tabular}{|c||c|c|c|c|c|c|c|}
\hline
$\mathbf{D}_1^{(1)}$ & 1 & 0 & 0 & 0 & 0 & 0 & 0 \tabularnewline
\hline
$\mathbf{D}_2^{(1)}$ & 0 & 0 & 1 & 0 & 0 & 0 & 0 \tabularnewline
\hline
$\mathbf{D}_3^{(1)}$ & 0 & 0 & 0 & 0 & 0 & 1 & 0 \tabularnewline
\hline
$\mathbf{D}_4^{(1)}$ & 0 & 1 & 0 & 0 & 0 & 0 & 1 \tabularnewline
\hline
$\mathbf{D}_5^{(1)}$ & 0 & 0 & 0 & 1 & 0 & 0 & 0 \tabularnewline
\hline
$\mathbf{D}_6^{(1)}$ & 0 & 0 & 0 & 0 & 1 & 0 & 1 \tabularnewline
\hline
$\mathbf{D}^{(1)}_7$ & 0 & 1 & 0 & 0 & 1 & 0 & 0 \tabularnewline
\hline
\end{tabular}
\end{center}
\caption{Table of available information at node $k=1$ when information diffusion is done via the \ac{TAS} algorithm.}
\label{Table}
\end{table}

\vskip 1.5mm iv) \textit{Aggregation phase}. To form the next packet to transmit, each node aggregates the information contained in $\mathbf{R}^{(k)}$, summing the available \textit{data sets} and merging the related \textit{tag vectors}. The merge is done as follows. When the aggregation phase takes place for the first time, each node $k$ initializes a \textit{temporary data set} $\mathbf{D}_T^{(k)}$ and the related $N$ elements \textit{temporary tag vector} $\mathbf{t}_T^{(k)}$ with the content of the first row of $\mathbf{R}^{(k)}$ and marks this row as already merged. Then, the node checks whether the next row contains some information that is already accounted for in $\mathbf{t}_T^{(k)}$. If not, it marks it as already merged and sums its corresponding \textit{data set} to $\mathbf{D}_T^{(k)}$ and updates $\mathbf{t}_T^{(k)}$. All rows are then progressively examined. Successive aggregation phases initialize $\mathbf{D}_T^{(k)}$ and $\mathbf{t}_T^{(k)}$ as the content of the first never merged row, starting the search from the first row.

\emph{Example 2}: Consider $\mathbf{R}^{(1)}$ reported in Table~\ref{Table}. Node $k=1$ has to compose the first \textit{message} that it should transmit. It starts from the first
row and initializes $\mathbf{D}_T^{(1)}$ and $\mathbf{t}_T^{(1)}$ with the content of the first row. It then marks the first row as already merged. The second row is then examined. As it contains only new information, with respect to the content of $\mathbf{t}_T^{(1)}$, its corresponding
\textit{data set} is added to $\mathbf{D}_T^{(1)}$ and its \textit{tag vector} is merged with $\mathbf{t}_T^{(1)}$, resulting in  $\mathbf{t}_T^{(1)}=(1, 0, 1, 0, 0, 0)$ and $\mathbf{D}_T^{(1)}=\mathbf{D}_1^{(1)}+\mathbf{D}_2^{(1)}$. The same happens for the third, fourth and fifth rows, that are then all marked as already merged. The sixth row contains, instead, information relative to node 7: Node 7 is already contributing to the current $\mathbf{t}_T^{(1)}$, thus, $\mathbf{D}_6^{(1)}$ is not added to $\mathbf{D}_T^{(1)}$. Afterwards, $\mathbf{D}_T^{(1)}$ and $\mathbf{t}_T^{(1)}$ are transmitted as definitive message, when all rows of the table have been traversed. When the next aggregation phase takes place, $\mathbf{D}_T^{(1)}$ and $\mathbf{t}_T^{(1)}$ are initialized as the content of the first row that has never been merged in the previous 
aggregation phases: In our example, this happens for the sixth row.


\vskip 1.5mm v) \textit{Transmission phase}. The message obtained at the end of the aggregation phase is broadcasted to all neighbor nodes. 

The information diffusion stops after a fixed number of transmission phases: On random networks the limit can be set equal to the diameter of the network (as would be the case for any flooding approach). 

\vskip 1.5mm vi) \textit{Wrap-up phase}. Once the information diffusion expires, the objective, for any node $k$, is the computation of \eqref{dataTAS}, which is then used to evaluate \eqref{zj}. This means finding a strategy to combine the rows in $\mathbf{R}^{(k)}$ to obtain the aggregated data in \eqref{dataTAS}. Two cases are possible: Either $\mathit{rank}\left(\mathbf{T}^{(k)}\right)=N$ and then a perfect reconstruction of \eqref{dataTAS} is possible, since each appearing term can be individually retrieved, or $\mathit{rank}\left(\mathbf{T}^{(k)}\right)<N$ and node $k$ will try to close as much as possible on \eqref{dataTAS}. This can be realized performing a linear combination of the rows of $\mathbf{R}^{(k)}$, aiming at maximizing the amount of data taken into account.

Each node $k$ will evaluate a linearly weighted sum $\mathbf{D}_F^{(k)}=\sum_r \widehat{b}_{k,r} \mathbf{D}_r^{(k)}$, where $\widehat{\mathbf{b}}_k$ is the solution of the following constrained optimization problem
\begin{align}\label{finalProblem}
\widehat{\mathbf{b}}_k&=\operatorname*{arg} \operatorname*{\max}_{\mathbf{b}_{k}} \mathbf{b}_{k} \mathbf{T}^{(k)} \mathbf{1},\\
\textrm{s.t. } 0\le \sum_{r}&b_{k,r}t_{r,i}^{(k)}\le1, \; \; i=1,\ldots,N. \label{constraints}
\end{align}
Here, $t_{r,i}^{(k)}$ are the elements of $\mathbf{T}^{(k)}$, with $r$ and $i$ denoting the row and column indexes. The solution of \eqref{finalProblem}-\eqref{constraints} is obtained by linear programming.

The term $c_{k,i}=\sum_{r}b_{k,r}t_{r,i}^{(k)}$ in \eqref{constraints} represents the weight of the quantities related to node $i$.
Since local quantities in \eqref{sj} cannot contribute more than once, to keep independence among all terms intervening in \eqref{sj}, then it must be $ 0\le c_{k,i} \le1$, that determines the constraints \eqref{constraints}.


\remark{The \ac{TAS} algorithm takes some inspiration from network coding techniques \cite{LiYeuCai:03,KoeMed:03, HoMedKoe:06}. However, the main difference is that each node does not need to decode, by means of Gaussian elimination, all the individual messages transmitted by the other nodes, but rather the decoding of their sum (possibly of an \emph{incomplete} sum) suffices.}

The performance of the \ac{TAS} algorithm will be investigated in Sections~\ref{sec:ThAnalysis} and~\ref{sec:SimRes}.

\subsection{Consensus Algorithm}
Given that the \ac{SPS} algorithm does not require the single terms appearing in \eqref{s0exp} and \eqref{sjexp} but rather their sum, a possibility to compute \eqref{s0exp} and \eqref{sjexp}, in a distributed way, is to launch an average consensus algorithm \cite{XiaBoy:04, XiaBoyLal:05, Xiao06, XiaBoyKim:07}, converging to \eqref{dataTAS}, as recently proposed in \cite{ZamKieBasPasDar:14}. For this information diffusion strategy, $\mathbf{R}^{(k)}$ is always composed of a single row, storing the consensus state vector. Further details can be found in the referenced papers \cite{XiaBoy:04, XiaBoyLal:05, Xiao06, XiaBoyKim:07,ZamKieBasPasDar:14}. Consensus algorithms will be considered in the numerical results section, anyway we will not put more emphasis since they showed a poor performance in terms of generated traffic load, as investigated in \cite{ZamKieBasPasDar:14}.

\section{Analysis of information diffusion truncation}\label{sec:Trunc}
In this section, the effect of truncation of information diffusion is discussed. The objective is to prove that consistent non-asymptotic confidence regions can still be computed via \ac{SPS}, at all nodes, even when the information diffusion process is stopped before each node has gathered all data.

To achieve this objective, the truncated expressions of \eqref{s0} and \eqref{sj} are provided first. Then, some other preliminary definitions and recalls are outlined. Last, a theorem closes the section.

Truncating the information diffusion algorithm entails that \eqref{s0}
and \eqref{sj} are estimated taking into account only the data actually
received by each node. Hence, at node $k$, the following quantities are evaluated from the available data
\begin{align}
\tilde{\mathbf{S}}_{k,0}(\mathbf{p})= & \sum_{i=1}^{N}c_{k,i}\boldsymbol\varphi_{i}\left(Y_{i}-\boldsymbol\varphi_{i}^{T}\mathbf{p}\right)\label{s0trunc}\\
\tilde{\mathbf{S}}_{k,j}(\mathbf{p})= & \sum_{i=1}^{N}c_{k,i}A_{j,i}\boldsymbol\varphi_{i}\left(Y_{i}-\boldsymbol\varphi_{i}^{T}\mathbf{p}\right),\label{sjtrunc}
\end{align}
where $j=1,\ldots,m-1$, and $c_{k,i}\in\left\{0,1\right\}$. The coefficients
$c_{k,i}$ reckon with the availability or absence of the $i$-th measurement, due to truncation, at node $k$.\footnote{The $c_{k,i}$ differ from the weighting coefficients appearing in \cite[Section~2.2]{CsaCamWey:12}. The difference is that, here, they depend on the measurement index $i$. This makes the two forms of weighting completely unrelated.}

Note that \eqref{s0trunc} is the set of normal equations that would be obtained
in a centralized context, considering a weighted least-squares estimator,
with a diagonal weight matrix $\mathbf{C}_{k}=\mathrm{diag}\left(c_{k,1},\dots,c_{k,N}\right)$.
Similarly, \eqref{sjtrunc} is the sign perturbed sum that would be
obtained when considering weighted least-squares. It will be shown that the confidence region,
obtained considering \eqref{s0trunc} and \eqref{sjtrunc} in \eqref{eq:Test},
is still a non-asymptotic confidence region. Reaching completion of the information diffusion algorithm entails that
the $c_{k,i}$ are all equal to one, 
thus ensuring equivalence with the centralized scenario. In case of truncation, instead, the $c_{k,i}$ fall in the interval $[0,1]$, their values depending on the applied information diffusion procedure: In case that the \ac{TAS} or a consensus approach are applied they might take any value in $[0,1]$, otherwise, with flooding, only 0 and 1 are possible values.

Taking the squared norms of \eqref{s0trunc} and \eqref{sjtrunc}, respectively named $\tilde{Z}_0(\mathbf{p})$ and $\tilde{Z}_j(\mathbf{p})$, for $j=1,\ldots,m-1$, allows to define the confidence region that is obtained at node $k$ when truncation occurs, that is,
\begin{align}
\widetilde{\mathbf{\Sigma}}_{q,k} & =\left\{ \mathbf{p}\in\mathbb{P} \left|\quad\sum_{j=1}^{m-1}\mathbb{I}(\tilde{Z}_{j}(\mathbf{p})-\tilde{Z}_{0}(\mathbf{p}))\geq q\right.\right\}.
\end{align}
In order to characterize the consistency of $\widetilde{\mathbf{\Sigma}}_{q,k}$, that relies on an incomplete set of measurements, it is necessary to recall some definitions taken from \cite{CsaCamWey:12}.
\begin{definition}[Symmetric Random Variables]\label{SymVar}
Given a probability space $(\Omega,\mathcal{F},\mathbbold{\Pi})$, $\Omega$ being the sample space, $\mathcal{F}$ the $\sigma$-algebra of events, and $\mathbbold{\Pi}$ the probability measure, a real (possibly $\mathbb{R}^d$-valued) \ac{RV} $X$ is said to be symmetric about the origin $0$ (possibly origin vector $\mathbf{0}$) if
\begin{align}
\forall A \in \mathcal{F}: \mathbb{P}(X\in A)=\mathbb{P}(-X\in A).
\end{align}
\end{definition}
The following property recalls \cite[Lemma~2]{CsaCamWey:12}.
\begin{Property}\label{property1}
Let $A, B_1,\ldots, B_k$ be \ac{i.i.d.} random signs. Then $A,A B_1, \ldots, AB_k$ are also \ac{i.i.d.} random signs.
\end{Property}
\begin{definition}[Uniformly Ordered Variables]\label{Def:UnifOrd}
A finite set of real-valued \acp{RV} $Z_0, Z_1, \ldots, Z_{m-1}$ is said to be uniformly ordered if for all permutations $i_0,i_1, \ldots,i_{m-1}$ of indexes $0,1,\ldots,m-1$, one has
\begin{align}
\mathbb{P}(Z_{i_0}<Z_{i_1}<\ldots<Z_{i_{m-1}})=\frac{1}{m!}.
\end{align}
\end{definition}
Definition \ref{Def:UnifOrd} states that 
all orderings are equiprobable. A direct consequence is that, for a set of uniformly ordered \acp{RV} $Z_0, Z_1, \ldots, Z_{m-1}$, each variable $Z_i$ takes any position in the ordering with probability $1/m$.

%
With the purpose to formulate Lemma~\ref{UnifOrderLemma}, introduced in the following, another few more considerations are needed. In this regard, let $h(Z_0,Z_1,\ldots,Z_{m-1}): \mathbb{R}^m\to \mathbb{N}_0^{m-1}$ be a function of $m$ real variables, with $\mathbb{N}_0^{m-1}$ denoting the set of naturals from 0 to $m-1$. The function provides a permutation $i_0,i_1,\ldots,i_{m-1}$, such that $Z_{i_0}\le Z_{i_1}\le \ldots \le Z_{i_{m-1}}$. In case of ties between input variables, the permutation is uniquely determined by applying the following rule. Suppose that
$n$ variables are tied: Thus $n!$ orderings are possible. Then $h(\cdot)$ provides a reordering choosing among the possible $n!$ with uniform distribution. Having premised this, when $h(\cdot)$ takes \acp{RV} as inputs, it can be considered as a discrete random variable with $m!$ possible outcomes, \emph{i.e.}, as many as the number of possible permutations of $m$ integers.
\begin{lemma}[Uniform Ordering Lemma]\label{UnifOrderLemma}
Let $Z_0,Z_1,\ldots,Z_{m-1}$ be real-valued, \ac{i.i.d.} \acp{RV}. Then they are uniformly ordered.
\end{lemma}
\begin{proof}
Consider $h(\cdot)$, as previously defined. Since $Z_0,Z_1,\ldots,Z_{m-1}$ are \ac{i.i.d.} the distribution of $h(Z_{i_0},Z_{i_1},\ldots,Z_{i_{m-1}})$ is the same for all permutations. Permutations are in number of $m!$, hence each of the outcome of $h(\cdot)$ has probability $1/m!$, since the mechanism, by which $h(\cdot)$ is defined, guarantees that all outcomes are equally possible. This is equivalent to saying that the variables are uniformly ordered.
\end{proof}
Lemma~\ref{UnifOrderLemma} is a generalization of \cite[Lemma~4]{CsaCamWey:12}, that does not hold for discrete \acp{RV}, to both continuous and discrete \acp{RV}. The need for this extension will appear in the proof of Theorem~\ref{Th1}.

Now, one can state the following theorem.
\begin{theorem} \label{Th1}
Under the assumption of measurement noises being symmetric \acp{RV} and independent across nodes, the confidence level with which the true parameter value $\mathring{\mathbf{p}}$ falls in the region $\widetilde{\boldsymbol\Sigma}_{q,k}$, yielded at node $k$, is
\begin{align}
\text{Prob}(\mathring{\mathbf{p}}\in \widetilde{\boldsymbol\Sigma}_{q,k})=1-\frac{q}{m},
\end{align}
for every $k=1,\ldots, N$.
\end{theorem}
\begin{proof}
Following a similar approach as in \cite{CsaCamWey:12},
the evaluation of \eqref{s0trunc} and \eqref{sjtrunc} for $\mathring{\mathbf{p}}$
gives
\begin{align}\label{s0p*}
\tilde{\mathbf{S}}_{k,0}(\mathring{\mathbf{p}})=\sum_{i=1}^{N}c_{k,i}\boldsymbol\varphi_{i}W_{i}
\end{align}
and
\begin{align}\label{sjp*}
\tilde{\mathbf{S}}_{k,j}(\mathring{\mathbf{p}})=\sum_{i=1}^{N}c_{k,i}A_{j,i}\boldsymbol\varphi_{i}W_{i},
\end{align}
with $j=1,\ldots,m-1$.
The truncation results in a
rescaling of measurement noise terms $W_{i}$, since it only depends on
the communication links effectively traversed during the information
diffusion phase. This rescaling preserves independence as well as
symmetry of noise distributions. 
Consider, further, that from \eqref{s0p*} and \eqref{sjp*}, one can derive
\begin{align}
\tilde{Z}_{k,0}(\mathring{\mathbf{p}}) & =\left\Vert \sum_{i=1}^{N}c_{k,i}\boldsymbol\varphi_{i}W_{i}\right\Vert _{2}^{2},
\end{align}
and
\begin{align}
\tilde{Z}_{k,j}(\mathring{\mathbf{p}}) & =\left\Vert \sum_{i=1}^{N}c_{k,i}A_{j,i}\boldsymbol\varphi_{i}W_{i}\right\Vert _{2}^{2}.
\end{align}
These last two expressions may be rewritten highlighting the independent
random measurement noise terms $W_{1},\ldots,W_{N}$, \emph{i.e.},
\begin{align}\label{Z0tilde}
\tilde{Z}_{k,0}(\mathring{\mathbf{p}}) & =f(c_{k,1}W_{1},\ldots,c_{k,N}W_{N}),
\end{align}
\begin{align}\label{Zjtilde}
\tilde{Z}_{k,j}(\mathring{\mathbf{p}}) & =f(c_{k,1}A_{j,1}W_{1},\ldots,c_{k,N}A_{j,N}W_{N}),
\end{align}
As already pointed out, each $c_{k,i}W_{i}$ has a symmetric distribution. By applying Lemma~1 from \cite{CsaCamWey:12} to the variables in the collection $\{c_{k,i}W_{i}\}_{i=1}^N$ and introducing the set of random signs $\left\{ B_{i}\right\} _{i=1}^{N}$ we can write $c_{k,i}W_{i}=B_{i}(B_{i}c_{k,i}W_{i})=B_{i}V_{i}$, where $B_i$ and $V_i=B_{i}c_{k,i}W_{i}$ are independent $\forall i$ \cite[Lemma~1]{CsaCamWey:12}. We can compact \eqref{Z0tilde} and \eqref{Zjtilde} in the single expression
\begin{align}
\tilde{Z}_{k,j}(\mathring{\mathbf{p}}) & =f(D_{j,1}V_{1},\ldots,D_{j,N}V_{N})
\end{align}
for $j=0,\ldots,m-1$, with $D_{0,i}\triangleq B_{i}$ and $D_{j,i}\triangleq A_{j,i}B_{i}$
for $j=1,\ldots,m-1$. The set of \acp{RV} $\left\{ D_{j,i}\right\} _{i=1,\, j=0}^{N,\, m-1}$ is
also a collection of \ac{i.i.d.} random signs, this deriving from Property~\ref{property1} applied to the \ac{i.i.d.} random signs $B_{i}$ and $\left\{ A_{j,i} \right\}_{j=1}^{m-1}$. Now fix a realization for $\left\{ V_{i}\right\} _{i=1}^{N}$, indicated as $\left\{ v_{i}\right\} _{i=1}^{N}$. Conditioning
on $\left\{ V_{i}\right\}_{i=1}^{N}=\left\{ v_{i}\right\} _{i=1}^{N}$,
$\left\{ \tilde{Z}_{k,j}(\mathring{\mathbf{p}}) \big| \left\{ V_{i}\right\}_{i=1}^{N}=\left\{ v_{i}\right\} _{i=1}^{N}  \right\}_{j=0}^{m-1}$ is a
collection of discrete, real-valued, and \ac{i.i.d.} \acp{RV}, since $\left\{ D_{j,i}\right\} _{i=1,\, j=0}^{N,\, m-1}$ is a collection of \ac{i.i.d.} random signs. Applying Lemma~\ref{UnifOrderLemma} to $\left\{ \tilde{Z}_{k,j}(\mathring{\mathbf{p}}) \big| \left\{ V_{i}\right\}_{i=1}^{N}=\left\{ v_{i}\right\} _{i=1}^{N} \right\}_{j=0}^{m-1}$ leads to the consideration that these variables are uniformly ordered.
This implies that the \ac{RV} $\tilde{Z}_{k,0}(\mathring{\mathbf{p}}) \big| \left\{ V_{i}\right\}_{i=1}^{N}=\left\{ v_{i}\right\} _{i=1}^{N}$ takes each position in the ordering with probability $1/m$. The conclusion is that it
is not among the $q$ largest $\tilde{Z}_{k,j}(\mathring{\mathbf{p}}) \big| \left\{ V_{i}\right\}_{i=1}^{N}=\left\{ v_{i}\right\} _{i=1}^{N}$,
$j=0,\ldots,m-1$, with probability $1-q/m$. Since this probability value is independent of the particular realization of $\left\{ V_{i}\right\}_{i=1}^{N}$, one can apply \cite[Lemma~3]{CsaCamWey:12} to say that $\text{Prob}(\mathring{\mathbf{p}}\in \widetilde{\boldsymbol\Sigma}_{q,k})=\text{Prob}\left(\tilde{Z}_{k,0}(\mathring{\mathbf{p}}) \text{ is not among the $q$ largest } \tilde{Z}_{k,j}(\mathring{\mathbf{p}})\right)=1-\frac{q}{m}$ is valid also when not conditioning on noise realizations. This concludes the proof.
\end{proof}
\begin{remark}
We introduced Lemma~\ref{UnifOrderLemma} to prove that \ac{i.i.d.} discrete \acp{RV} are uniformly ordered. From \cite[Lemma~4]{CsaCamWey:12}, one can draw this conclusion only for continuous variables. 
The here presented Lemma~\ref{UnifOrderLemma} generalizes \cite[Lemma~4]{CsaCamWey:12}. This copes with the discrete \acp{RV}, that are appearing when noise realizations are fixed, as done in the proof.
\end{remark}
\begin{remark}\label{RemarkShape}
When at sensor node $k$ there is only a single non-zero coefficient, $c_{k,k}=1$, meaning that truncation in information diffusion occurred before node $k$ could gain knowledge about any other sensor than itself, then, the $m\times N$ matrix formed by all random signs $A_{j,i}$ participating in the confidence region computation at node $k$ has only one column filled with values $\{-1,1\}$, while all the remaining ones are filled with zeros.
Its rank is hence equal to 1 and the norms $\tilde{Z}_{k,j}(\mathbf{p})$ are all equal independently of $j$ and for any value of $\mathbf{p}$. This is certainly the case for which the highest number of ties occurs, nevertheless, choosing at random for the reordering, yields a random confidence region, covering a percentage equal to $1-q/m$ of the initial search space. The computed confidence region keeps again the same level of confidence $1-\frac{q}{m}$, as stated by Theorem~\ref{Th1}. This observation gives an insight on the reason why the shape of confidence regions is affected by information availability. 
\end{remark}

\section{Traffic load on various network topologies} \label{sec:ThAnalysis}
For a fair comparison of different information diffusion strategies, the network traffic burden has to be characterized. The algorithms are compared on specific topologies, such as random trees, with binary trees as a special case, and clustered networks, that are the most commonly used in practical applications \cite{VerDarMazCon:B08}. In Section~\ref{sec:SimRes}, completely random networks will also be considered. 

Before entering into the details of our analytical investigation, let us recall that $d_{\text{TAS}}$ and $d_{\text{MF}}$, respectively given by \eqref{dTAS} and \eqref{dMF}, denote the numbers of real-valued scalars that a single data is composed of when the \ac{TAS} or the \ac{MF} algorithm are considered.

The remainder of this section is divided into as many subsections as the considered topologies.

\begin{figure}[t]
\centering \includegraphics[width=1\columnwidth]{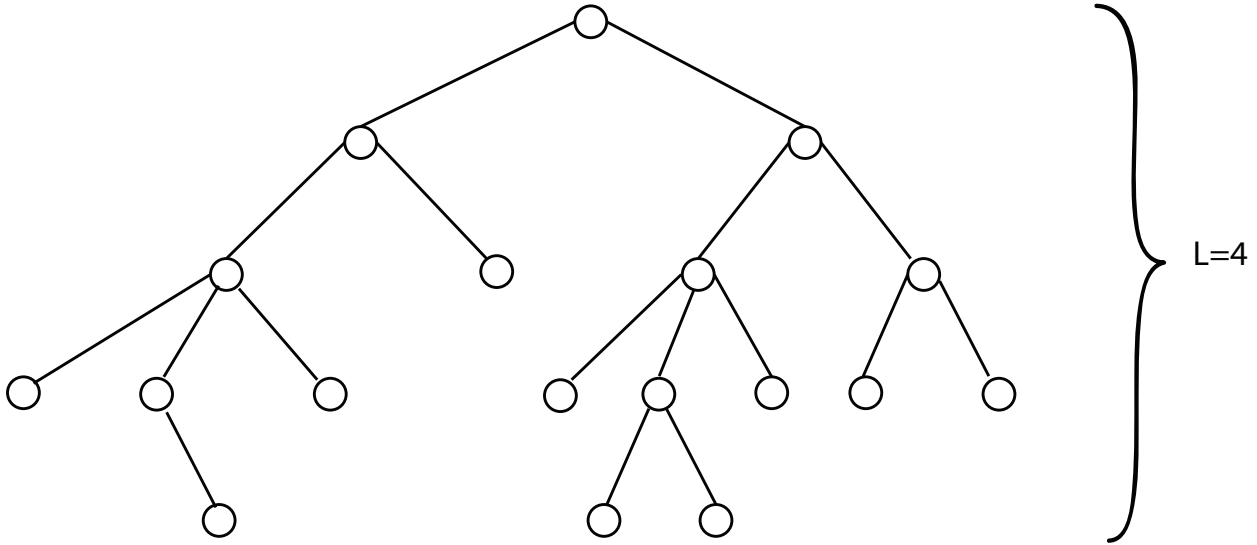} \caption{A random tree topology. Some of the random variables, that describe it, take values $\Lambda(0)=1,\, \Lambda(1)=2,\, \Lambda(2)=4,\, \bar{\Lambda}(2)=1, \, \Lambda(3)=8,\, \bar{\Lambda}(3)=6$, etc.}
\label{Fig:RandomTree}
\end{figure}

\subsection{Random Tree Topology}\label{subsec:RandomTree}
Consider a random tree topology, \emph{i.e.}, a tree where each node has a random number of sons, possibly zero. 
The number of nodes forming the network is considered equal to $N$. The levels in the tree are indicated by $\ell$, with $\ell\geq0$. $L$ denotes the lowest level in the tree.
The set of nodes in the $\ell$-th level of the tree is denoted as $\mathbb{L}_\ell$, having cardinality $\Lambda(\ell)$, which is a \ac{RV}. Nevertheless, $\Lambda(0)=1$ and is not random, since the tree is single rooted. Moreover, the set of nodes, in the $\ell$-th level in the tree, that are not parents to any nodes, is denoted by $\bar{\mathbb{L}}_\ell$ and its cardinality is a \ac{RV} denoted by $\bar{\Lambda}(\ell)$. 

\subsubsection{\ac{TAS} algorithm}
The \ac{TAS} algorithm of Section~\ref{Subsec:TAS} does not assume any ordering in the network on which it should run. On a random tree, however, it is possible to simplify it making nodes transmit much less frequently than required on an unstructured random topology.

During each transmission phase a single level of the tree is active. Only nodes in this level can transmit. Starting from level $L$, each node in $\mathbb{L}_L$ has to broadcast its own local data. Then, parent nodes distil and aggregate the received quantities with
their own ones and broadcast. The process is repeated until the tree root is reached. The tree is then traveled backwards, from Level 0 to Level $L$, making the complete sum available to all
nodes. This way of operating ensures that an exact retrieval of the entire sum is possible and that no truncation occurs, if the procedure is completed. Nodes participate only in the (at most two) rounds of transmission involving the level they belong to.
The number of data that must be transmitted when employing the \ac{TAS} algorithm is a discrete \ac{RV} 
given by
\begin{align}\label{dataTASTree}
N_{\text{TAS}}^{\text{RT}}=\sum_{\ell=0}^{L} \Lambda(\ell) d_{\text{TAS}} +\sum_{\ell=1}^{L-1}\Lambda(\ell) d_{\text{TAS}} -\sum_{\ell=1}^{L-1}\bar{\Lambda}(\ell) d_{\text{TAS}}.
\end{align}
$N_{\text{TAS}}^{\text{RT}}$ consists of the number of data transmitted when traversing the tree from level $L$ to the root, included, plus the amount of data required by the backwards travel. The last term in \eqref{dataTASTree} is related to nodes without sons, which do not transmit anything when the travel backwards is performed.

\subsubsection{\ac{MF} algorithm}
For the \ac{MF}, one instead gets,
\begin{align}\label{dataMFTree}
N_\text{MF}^{\text{RT}}&=\Lambda(L) d_{\text{MF}}+(\Lambda(L)+\Lambda(L-1))d_{\text{MF}}+\ldots+\sum_{\ell=0}^L \Lambda(\ell) d_{\text{MF}} +N\left(\Lambda(1)-\bar{\Lambda}(1)\right) d_{\text{MF}} \nonumber\\
&\phantom{aaaaa}-\sum_{\ell=2}^L \Lambda(\ell) d_{\text{MF}}-\left(\Lambda(1)-\bar{\Lambda}(1)\right)d_{\text{MF}} +\ldots+ N\left(\Lambda(L-1)-\bar{\Lambda}(L-1)\right) d_{\text{MF}} \nonumber\\
&\phantom{aaaaa}- \Lambda(L) d_{\text{MF}}- \left(\Lambda(L-1)-\bar{\Lambda}(L-1)\right)d_{\text{MF}}\nonumber\\
&=\Lambda(L) d_{\text{MF}}\!+\!(\Lambda(L)\!+\!\Lambda(L-1))d_{\text{MF}}\!+\!\ldots\!+\!\sum_{\ell=0}^L \Lambda(\ell) d_{\text{MF}}\! +\! (N-1)\left(\Lambda(1)\!-\!\bar{\Lambda}(1)\right) d_{\text{MF}} \nonumber\\
&\phantom{aaaaa}-\sum_{\ell=2}^L \Lambda(\ell) d_{\text{MF}}+\ldots+(N-1)\left(\Lambda(L-1)-\bar{\Lambda}(L-1)\right) d_{\text{MF}} - \Lambda(L) d_{\text{MF}}\nonumber\\
 &=\sum_{\ell=0}^L \Lambda(\ell) d_{\text{MF}}+\Lambda(L)d_{\text{MF}}+N\sum_{\ell=1}^{L-1}\left(\Lambda(\ell)-\bar{\Lambda}(\ell)\right) d_{\text{MF}} +\sum_{\ell=1}^{L-1} \bar{\Lambda}(\ell) d_{\text{MF}},
\end{align}
where $\Lambda(L) d_{\text{MF}}+(\Lambda(L)+\Lambda(L-1))d_{\text{MF}}+\ldots+\sum_{\ell=0}^L \Lambda(\ell) d_{\text{MF}}$ is the amount of data transmitted in the forward travel, $N\left(\Lambda(1)-\bar{\Lambda}(1)\right) d_{\text{MF}}$ is the amount of data (proportional to $N$) that nodes, with sons, in level 1 would transmit when the tree is traveled backwards, as if no forward travel was ever performed, and $\sum_{\ell=2}^L \Lambda(\ell) d_{\text{MF}}+\left(\Lambda(1)-\bar{\Lambda}(1)\right)d_{\text{MF}}$ is the amount of data that has to be subtracted from the previous one, since these data have already been transmitted in the forward travel. The other terms can be similarly explained.

\subsubsection{Comparison}\ac{TAS} is more efficient than \ac{MF} when $N_\text{TAS}^{\text{RT}} <N_\text{MF}^{\text{RT}}$, i.e., when
\begin{align}
&\sum_{\ell=0}^L \Lambda(\ell) (d_{\text{TAS}}-d_{\text{MF}}) -\sum_{\ell=1}^{L-1} \bar{\Lambda}(\ell) (d_{\text{TAS}}+d_{\text{MF}}) - \Lambda(L) d_{\text{MF}} \nonumber\\
&\phantom{aaaaaaaaaaaaaaa}- N \sum_{\ell=1}^{L-1}\left( \Lambda(\ell) - \bar{\Lambda}(\ell) \right) d_{\text{MF}}  + \sum_{\ell=1}^{L-1} \Lambda(\ell) d_{\text{TAS}} < 0\nonumber\\
&N(d_{\text{TAS}}-d_{\text{MF}})-\Lambda(L) d_{\text{MF}}+(N-\Lambda(L)-1)d_{\text{TAS}}\nonumber\\
&\phantom{aaaaaaaaaaaaaaa}-N(N-\Lambda(L)-1)d_{\text{MF}} + \sum_{\ell=1}^{L-1} \bar{\Lambda}(\ell)\left((N-1)d_{\text{MF}}-d_{\text{TAS}}\right) < 0,
\end{align}
that is,
\begin{align}\label{randcomp}
\left(d_{\text{MF}}+d_{\text{TAS}}-Nd_{\text{MF}}\right)\Lambda(L) &> N(d_{\text{TAS}}-d_{\text{MF}})+(N-1)d_{\text{TAS}}-N(N-1)d_{\text{MF}} \nonumber\\
&\phantom{aaaaaaaaa}+ \sum_{\ell=1}^{L-1} \bar{\Lambda}(\ell) \left((N-1)d_{\text{MF}}-d_{\text{TAS}}\right).
\end{align}

In case $N$ is finite, \eqref{randcomp} is satisfied with a probability that is not easily evaluated. Nevertheless, an asymptotic consideration can be done. Firstly, when $N \rightarrow \infty$, we assume that $L \rightarrow \infty$. This is precisely the case when the area on which nodes are deployed is increasing with $N$, due to coverage extension purposes, or if the communication range $d_{\text{comm}}$ is diminishing with $N$, due to interference mitigation purposes. Moreover, one has $\sum_{\ell=1}^{L-1} \bar{\Lambda}(\ell) < N-\Lambda(L)-L$, since the number of nodes without sons, in all levels except 0 and $L$, cannot exceed the total number of nodes deprived of the number of nodes at level $L$ and of at least one node for each of the $L$ levels from 0 to $L-1$ (if no nodes are present in a level there cannot be any further levels). Then, passing to the limit for $N\to +\infty$, if
\begin{align}
\lim_{N\to \infty} \left\{ N( \right. & d_{\text{TAS}} -d_{\text{MF}}) +(N-1)d_{\text{TAS}}-N(N-1)d_{\text{MF}} \nonumber\\
&\left. + \left((N-1)d_{\text{MF}}-d_{\text{TAS}}\right) \left(N-\Lambda(L)-L\right)-\left(d_{\text{MF}}+d_{\text{TAS}}-Nd_{\text{MF}}\right) \Lambda(L) \right\} < 0
\end{align}
holds, \emph{i.e.},
\begin{align}\label{lim}
\lim_{N\to \infty} (d_{\text{TAS}}-d_{\text{MF}})+\frac{(L-1)}{N}d_{\text{TAS}}-\frac{L(N-1)}{N}d_{\text{MF}}<0,
\end{align}
then also \eqref{randcomp} holds asymptotically.
Since $L$ also goes to $+\infty$,
\eqref{lim} is verified and thus \eqref{randcomp} holds for all values of $\Lambda(L)$, hence with probability 1. Moreover, this is true for all values of the problem dimensions, \emph{i.e.}, $n_p$ and $m$.

\subsection{Binary Tree Topology}
A deterministic \emph{complete binary} tree topology, that is a tree where each node has exactly two sons apart from nodes in level $L$ that do not have any, is now considered. Assuming that the binary tree consists of $L$ levels entails $N=2^{L+1}-1$.

\subsubsection{\ac{TAS} algorithm}For the \ac{TAS} algorithm, the total number of required data communications is deduced from \eqref{dataTASTree}
\begin{align}
N_\text{TAS}^{\text{BT}} & = \left(2^{L}+2^{L-1}+\ldots+2^{1}+1+2^1+\ldots+2^{L-1}\right)d_{\text{TAS}}\nonumber\\
 & =\left(2^{L+1}-2+\frac{1}{2}\left(2^{L+1}-2\right)\right)d_{\text{TAS}} =\left(\frac{3N-3}{2}\right)d_{\text{TAS}}.
\end{align}

\subsubsection{\ac{MF} algorithm} The total number of data communications required by the \ac{MF} algorithm is deduced from \eqref{dataMFTree}
\begin{align}
N_\text{MF}^{\text{BT}}&=\left( 2^{L+1}-1 \right)d_{\text{MF}}+2^Ld_{\text{MF}} +\left( 2^{L+1}-1 \right)\sum_{i=1}^{L-1}2^i d_{\text{MF}} =\frac{N^2+1}{2}d_{\text{MF}}.
\end{align}

\begin{figure}[t]
\psfrag{N}[c][c][0.8]{$N^*$}
\psfrag{np}[c][c][0.8]{$n_p$}
\psfrag{mmmmmm1}[c][c][0.8]{$m=10$}
\psfrag{mmmmmm2}[c][c][0.8]{$m=20$}
\psfrag{mmmmmm3}[c][c][0.8]{$m=40$}
\centering \includegraphics[width=1\columnwidth, height=3 in]{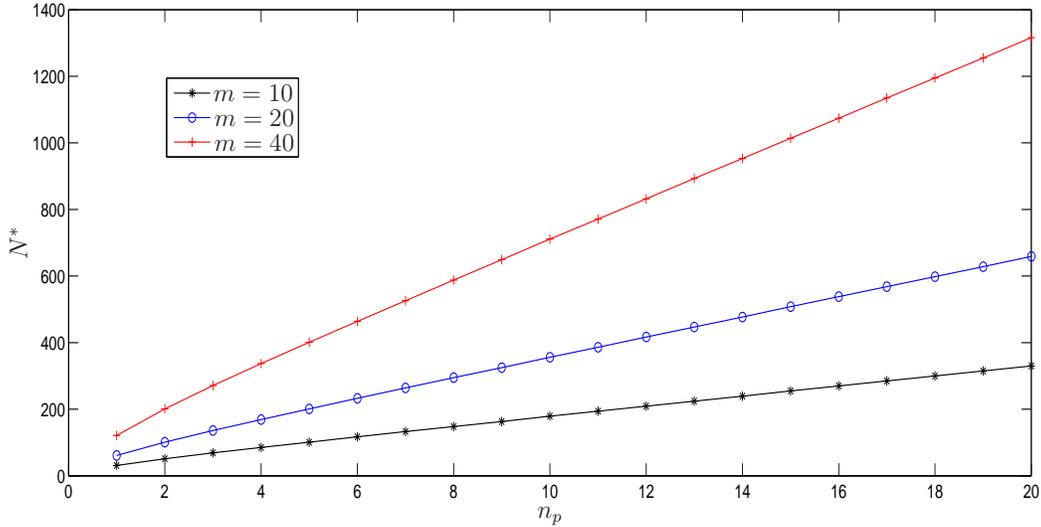} \caption{Critical value $N^*$, as a function of $n_p$, on binary trees, for several values of $m$.}
\label{Fig:CriticalN}
\end{figure}

\subsubsection{Comparison} On a binary tree, \ac{TAS} is more efficient than \ac{MF} when
\begin{align}
& \frac{3N-3}{2}d_{\text{TAS}} < \frac{N^2+1}{2} d_{\text{MF}}.
\end{align}
Using \eqref{dTAS} and \eqref{dMF} one obtains the following condition
\begin{align} \label{eq:comparison_binary_tree}
\left(N^2+1\right)K_1-3N+3 > 0,
\end{align}
where $K_1=\frac{n_p+1}{\left(n_p+n_p\frac{n_p+1}{2}\right)m}$. For sufficiently large $N$, \eqref{eq:comparison_binary_tree} is always satisfied, disregarding $n_p$ and $m$.
Moreover, and unlike in the random tree case, given $n_p$ and $m$, it is possible to derive the value
\begin{align}\label{critValue}
N^*=\frac{ 3+\sqrt{ 9-4K_1 \left(3+ K_1 \right)} } { 2K_1 },
\end{align}
for which \ac{TAS} is more efficient than \ac{MF}. Fig.~\ref{Fig:CriticalN} represents  $N^*$ as a function of $n_p$, considering $m=10, 20, 40$. The behaviour is not exactly linear, as it can be easily verified by derivation of \eqref{critValue}, but rapidly tends to be such: In fact, when $n_p$ grows large, $K_1\approx \frac{2}{n_p m}$ and $N^*\approx \frac{n_p m}{4}\left[ 3+\sqrt{9-\frac{8}{n_p m}\left( 3+ \frac{2}{n_p m} \right)} \;\right]\approx \frac{3}{2}n_p m$.


\subsection{Clustered Topology}
Consider a clustered network, formed by $N$ nodes, structured on a single level of hierarchy (see Fig. \ref{Fig:ClusterRing}). The network is hence assumed to be divided in $n_c$ clusters. The $i$-th cluster comprises a random number of nodes $N_i^{\text{c}}$, including the clusterhead, that is the special node responsible for aggregating the local data of its sons. The subnetwork formed by clusterheads is considered to be fully connected: Clusterheads can directly communicate to one another. Moreover, each node in
a cluster is assumed to directly communicate with its clusterhead (and vice-versa).

\begin{figure}[t]
\centering \includegraphics[width=0.7\columnwidth, height=2.8 in]{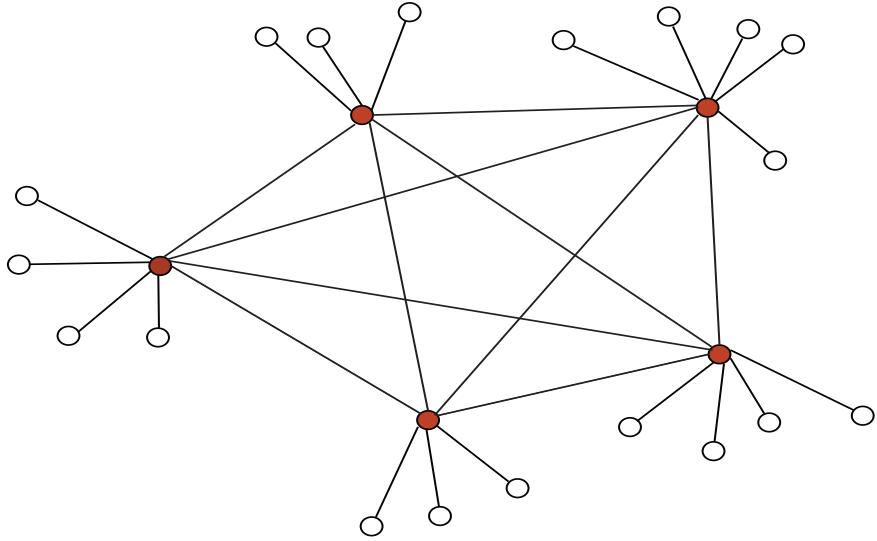} 
\captionsetup{justification=centering} \caption{A clustered topology. Clusterheads are indicated in red.}
\label{Fig:ClusterRing}
\end{figure}

\subsubsection{\ac{TAS} algorithm} On this topology, the \ac{TAS} algorithm transmission phases can be organized as follows. At the beginning, all nodes, with the exception of clusterheads, transmit their local data to the clusterheads. Then each clusterhead aggregates the
local data of all nodes in its cluster. 
Successively, clusterheads transmit to all other clusterheads their aggregated data. Since the network of clusterheads is fully connected, a single broadcast
transmission for each of the clusterhead suffices for all clusterheads being capable to construct the completely aggregated data. The
amount of scalar data, that has to be transmitted, is thus
\begin{align}
N^{\text{cc}}_\text{TAS} & =\left((N-n_c)+n_c+n_c\right)d_{\text{TAS}}\nonumber\\
&=(N+n_c)d_{\text{TAS}}.
\end{align}
This accounts for the initial $N-n_c$ transmissions and the
subsequent actions of clusterheads, that should broadcast to each
other the partially aggregated data and then broadcast, towards
nodes forming their cluster, the completely aggregated data. 

\subsubsection{\ac{MF} algorithm} All nodes in a cluster can overhear broadcast transmissions operated by the corresponding clusterhead. Therefore, the amount of data to be transmitted when employing the \ac{MF} algorithm is
\begin{align}
N^{\text{cc}}_\text{MF} & =\left(\left(N-n_c\right)+N+\left(n_c-1\right)N\right)d_{\text{MF}}\nonumber\\
&=\left(N-n_c+n_cN\right)d_{\text{MF}}.
\end{align}
This is because all nodes, apart from clusterheads, initially transmit their local information to clusterheads, giving rise to $(N-n_c)d_{\text{MF}}$ transmitted scalar data. Then clusterheads broadcast
the received data and their own, this forming a total flow of $Nd_{\text{MF}}$
scalar data. At this point, all nodes in each cluster are completely informed about data related to their respective cluster. Finally, there is a backwards transmission during which each clusterhead
is transmitting towards its cluster all the $Nd_{\text{MF}}$ scalar data except
the ones that it previously transmitted, this being equivalent
to further $\left(n_c-1\right)Nd_{\text{MF}}$ transmitted scalars, composed
of $n_c$ clusterheads transmitting not $N$, but $(N-N_{i}^{\text{c}})d_{\text{MF}}$ scalar data, \emph{i.e.},
a total of $\sum_{i=1}^{n_c}\left(N-N_{i}^{\text{c}}\right)d_{\text{MF}}=\left(n_cN-N\right)d_{\text{MF}}$. 

\subsubsection{Comparison} \ac{TAS} is better than \ac{MF} when $N_{\text{TAS}}^{\text{cc}}<N_{\text{MF}}^{\text{cc}}$, \emph{i.e.}, when
\begin{align}
\left(N-n_c+n_c N\right) d_{\text{MF}} -\left(N+n_c\right) d_{\text{TAS}} &> 0 \nonumber \\
\left(1+\frac{n_c(N-2)}{N+n_c}\right)\frac{d_{\text{MF}}}{d_{\text{TAS}}} &> 1.
\end{align}
Here $n_c$ is the degree of freedom, in lieu of $L$ in the tree topologies. Assuming that, due to coverage extension or interference mitigation purposes, $n_c$ grows to $\infty$ with $N$ going to $\infty$, one has
\begin{align}
\lim_{N\to +\infty} \left(1+\frac{n_cN-2n_c}{N+n_c}\right)\frac{d_{\text{MF}}}{d_{\text{TAS}}} =\lim_{N\to +\infty} n_c\frac{d_{\text{MF}}}{d_{\text{TAS}}} =+\infty,
\end{align}
independently on $n_p$ and $m$. Thus, \ac{TAS} is asymptotically better than \ac{MF}.

\begin{figure}[t]
\psfrag{p1}[c][c][0.8]{$p_1$}
\psfrag{p2}[c][c][0.8]{$p_2$}
\psfrag{p3}[c][c][0.8]{$p_3$}
\centering \includegraphics[width=\columnwidth , height=3 in]{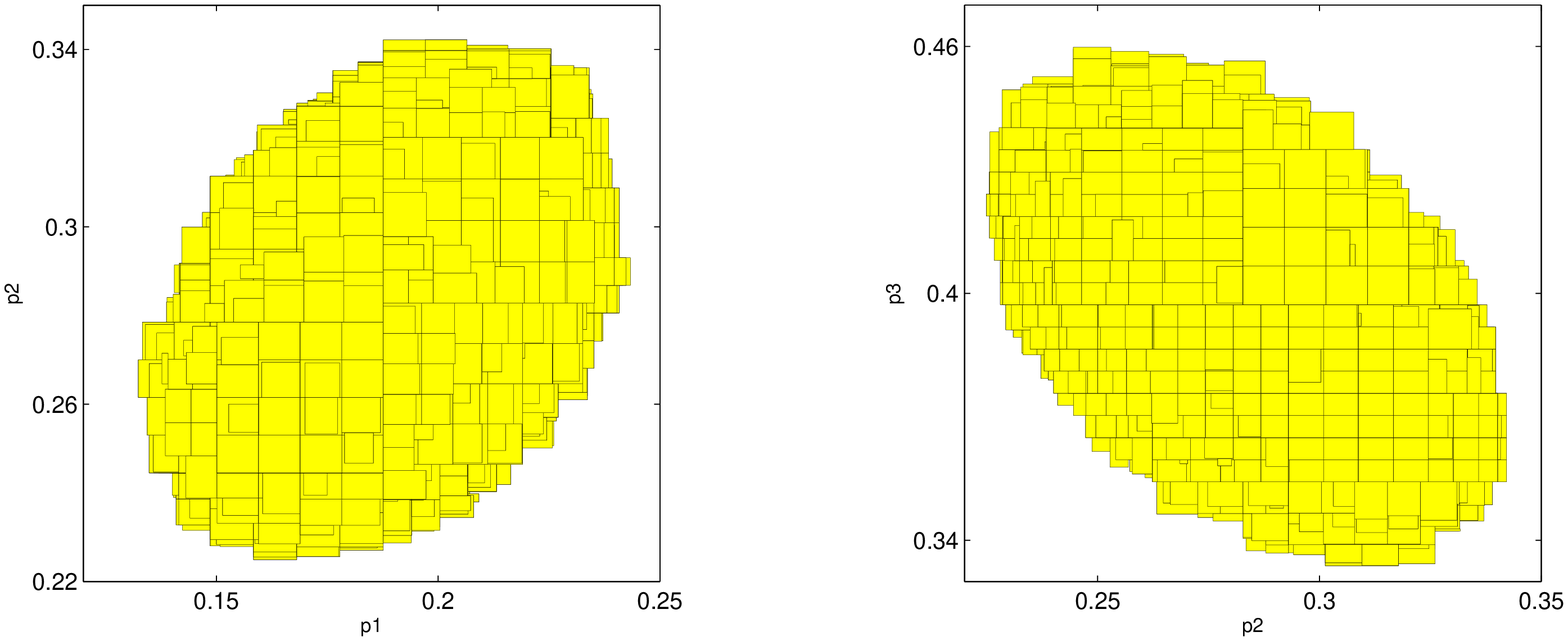} \caption{Projections of the 90\% confidence region computed at node 1 after 4 consensus iterations. A random unstructured network of 100 nodes is considered.}
\label{Fig:Cons4iter}
\end{figure}

\begin{figure}[t]
\psfrag{p1}[c][c][0.8]{$p_1$}
\psfrag{p2}[c][c][0.8]{$p_2$}
\psfrag{p3}[c][c][0.8]{$p_3$}
\centering \includegraphics[width=\columnwidth , height=3 in]{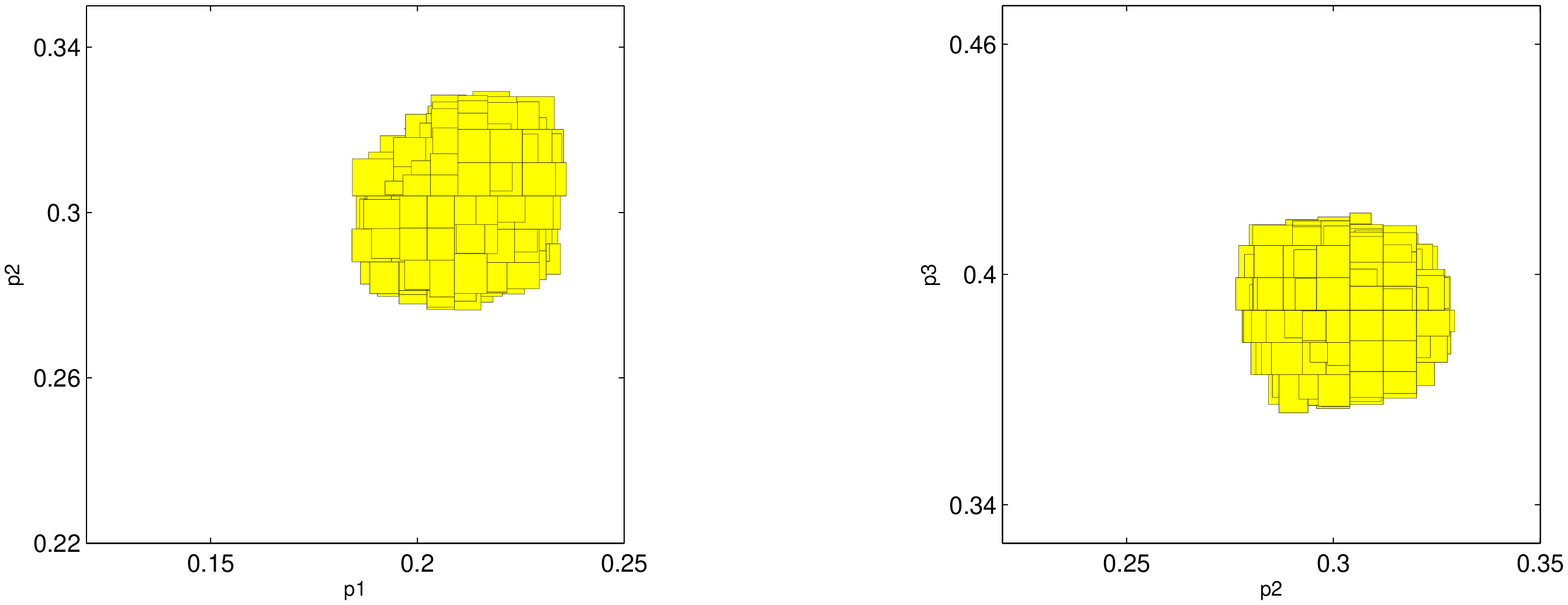} \caption{Projections of the 90\% confidence region computed at node 1 after 30 consensus iterations. A random unstructured network of 100 nodes is considered.}
\label{Fig:Cons30iter}
\end{figure}

\begin{figure}[t]
\psfrag{NumberofNodes}[c][c][0.8]{Number of Nodes}
\psfrag{Rate}[c][c][0.8]{Percentage of network realizations favorable to \ac{TAS}}
\psfrag{nppppppp2}[c][c][0.8]{$n_p=2$}
\psfrag{nppppppp3}[c][c][0.8]{$n_p=3$}
\psfrag{nppppppp4}[c][c][0.8]{$n_p=4$}
\psfrag{nppppppp5}[c][c][0.8]{$n_p=5$}
\centering \includegraphics[width=\columnwidth , height=3 in]{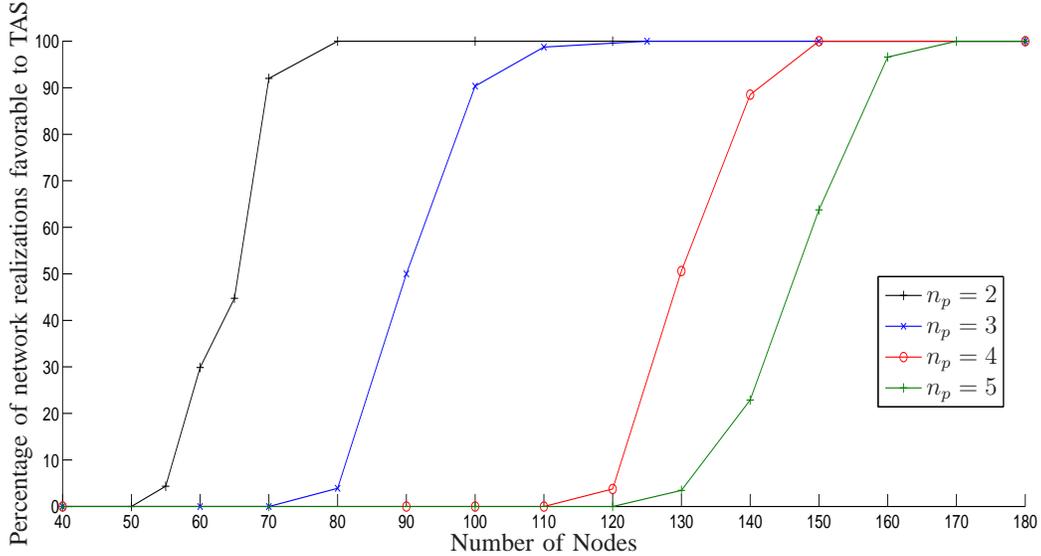} \caption{Percentage of network realizations favorable to \ac{TAS}, in terms of required data exchanges, compared to \ac{MF}, as a function of the number of nodes forming a random tree topology for different values of $n_p$. 100 random tree realizations are considered for each value of $N$.}
\label{Fig:SuccessRateRandomTree}
\end{figure}

\section{Simulation Results}\label{sec:SimRes}
In this section, all simulations results have been obtained assuming lossless links while confidence regions are evaluated with the interval analysis techniques described in \cite{KieWal:14}.\footnote{These techniques allow for the computation of tight outer approximations of confidence regions via contraction of the initial search space. The contraction halting criterion may be set such that single box outer approximations are obtained, instead of multiple boxes outer approximations. For the sake of simplicity and with abuse of terminology, in the remainder, `confidence regions' is used in lieu of `outer approximations of confidence regions', if not otherwise specified.} The Intlab package \cite{Ru99a} is employed for intervals computations. 

We start with a numerical investigation of the effect of truncation in information diffusion on the shape of the confidence region. To this purpose, we instantiate a random unstructured network of $N=100$ nodes, uniformly distributed over a unit area, and consider a true parameter value $\mathring{\mathbf{p}}=[p_1, p_2, p_3]=[0.2, 0.3, 0.4]$. The inter node communication range is set to $d_{\text{comm}}=\sqrt{\frac{\log_2 N}{2N}}$. According to \cite{GupKum:00}, this range guarantees almost sure connectivity of a network of $N$ nodes, deployed on a finite area. A truncated Metropolis consensus algorithm \cite{XiaBoy:04,Xiao06,ZamKieBasPasDar:14} is considered for the distributed computation of confidence regions. Similar results may be obtained also for the other information diffusion strategies, considered in Section~\ref{sec:InfoDiffAlg}. Figs.~\ref{Fig:Cons4iter} and \ref{Fig:Cons30iter} show the confidence region computed at node 1 after 4 and 30 iterations, respectively. The reduction in terms of volume is quite evident in the second case, while we underline that the confidence level is the same.
  
In order to compare the \ac{TAS} and the \ac{MF} algorithms, we consider random trees, clustered networks, and random unstructured topologies, for the same order of magnitude in terms of number of nodes.

For what concerns the analysis on random trees, we build a spanning tree on top of a random unstructured network, setting $d_{\text{comm}}$ as earlier done. For each $N$ (see the horizontal axis in Fig.~\ref{Fig:SuccessRateRandomTree}), 100 connected network realizations are instantiated. \ac{TAS} and \ac{MF} are compared in terms of the required number of data to be transmitted. The success rate of \ac{TAS} is the percentage of network realizations that proved favorable to \ac{TAS} and it is shown in Fig.~\ref{Fig:SuccessRateRandomTree} as a function of $N$, for several $n_p$ values. It can be noticed that, as foreseen in the theoretical analysis in Section~\ref{sec:ThAnalysis}, there always exists a threshold value of $N$, depending on $n_p$, above which the \ac{TAS} outperforms the \ac{MF} algorithm, \emph{i.e.}, the percentage closes to 100.

We now investigate the trade-off between the confidence region volume and the amount of per node transmitted data. 
Fig.\ref{Fig:AvgVolRandTree} shows the average volume of the $90$\% confidence region as a function of the average amount of data that need to be communicated by each node. The volume and data amount are averaged across all nodes and across 100 random tree realizations, while simulation parameters are set to $n_p=2$, $q=1$, $N=100$ and $m=10$. Fig.~\ref{Fig:AvgVolRandTree} allows to know which is the amount of data that need to be transmitted by each node on average to obtain a given confidence region average volume. Each pair of coordinates corresponds to one transmission round. It is easily seen that the \ac{TAS} algorithm outperforms the \ac{MF} for meaningful small volume values, in terms of the amount of per node transmitted data. 
 
\begin{figure}[t]
\psfrag{AverageVolume}[c][c][0.8]{Average Volume}
\psfrag{AverageAmountofTransmittedScalars}[c][c][0.8]{Average amount of per node transmitted scalars}
\psfrag{dataaa1}[c][c][0.6]{\ac{MF}}
\psfrag{dataaa2}[l][l][0.6]{\ac{TAS}}
\centering \includegraphics[width=\columnwidth , height=3 in]{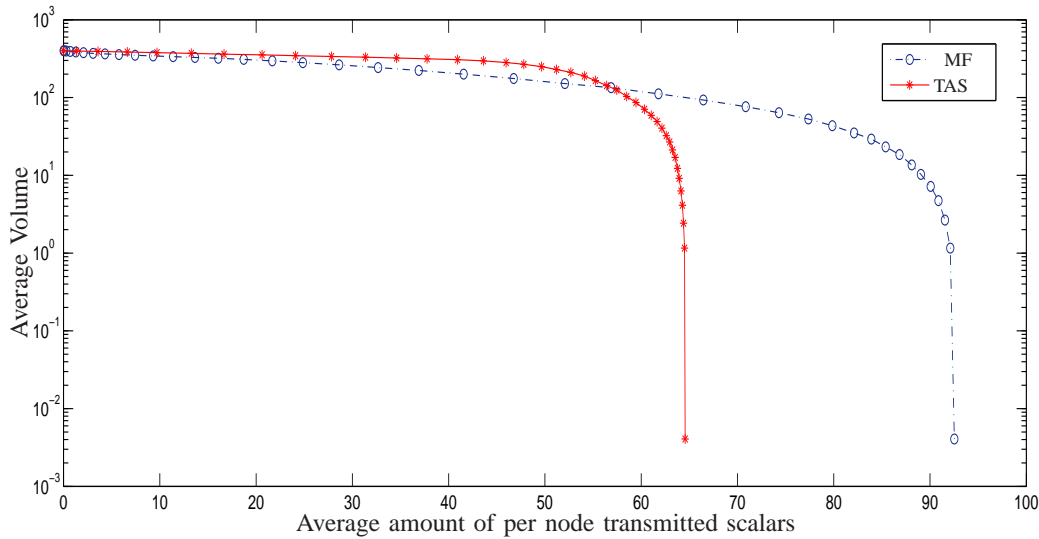} \caption{Average volume, across nodes and 100 random tree realizations, of the 90\% confidence region. Simulation parameters are set to $N=100, n_p=2, q=1$, and $m=10$.}
\label{Fig:AvgVolRandTree}
\end{figure}

\begin{figure}[t]
\psfrag{AverageVolume}[c][c][0.8]{Average Volume}
\psfrag{AverageAmountofTransmittedScalars}[c][c][0.8]{Average amount of per node transmitted scalars}
\psfrag{dataaa1}[c][c][0.6]{\ac{MF}}
\psfrag{dataaa2}[l][l][0.6]{\ac{TAS}}
\centering \includegraphics[width=\columnwidth , height=3 in]{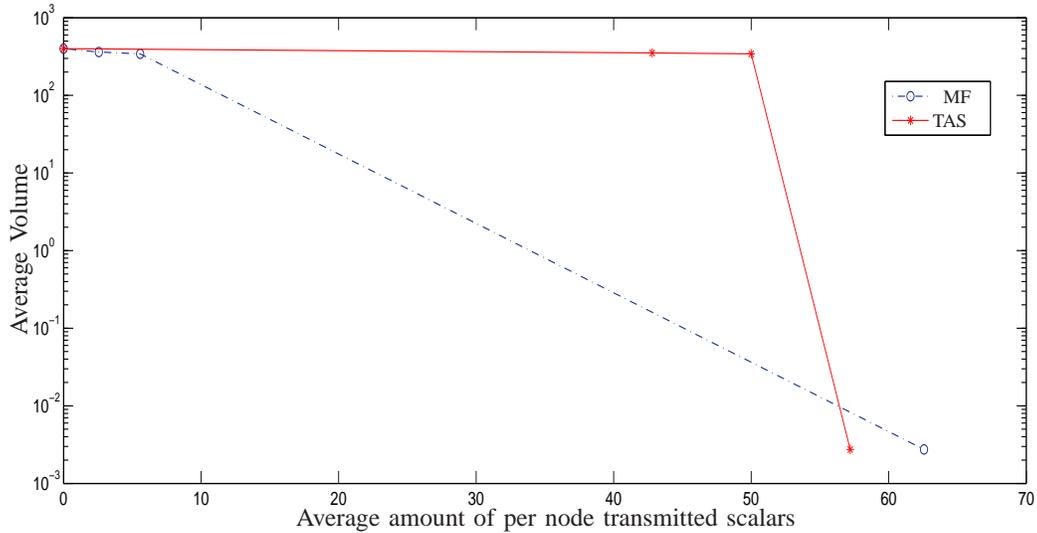} \caption{Average volume, across nodes and 100 clustered network realizations, of the 90\% confidence region. Simulation parameters are set to $n_p=2, q=1, n_c=20$, and $m=10$.}
\label{Fig:AvgVolClusteredNetworks}
\end{figure}

\begin{figure}[t]
\psfrag{AverageVolume}[c][c][0.8]{Average Volume}
\psfrag{AverageAmountofTransmittedScalars}[c][c][0.8]{Average amount of per node transmitted scalars}
\psfrag{data1}[c][c][0.6]{\ac{MF}}
\psfrag{data2}[l][l][0.6]{\ac{TAS}}
\psfrag{MetropolisConsensusssssss}[l][l][0.6]{Metropolis Consensus}
\psfrag{data4}[l][l][0.6]{Perron Consensus}
\centering \includegraphics[width=\columnwidth , height=3 in]{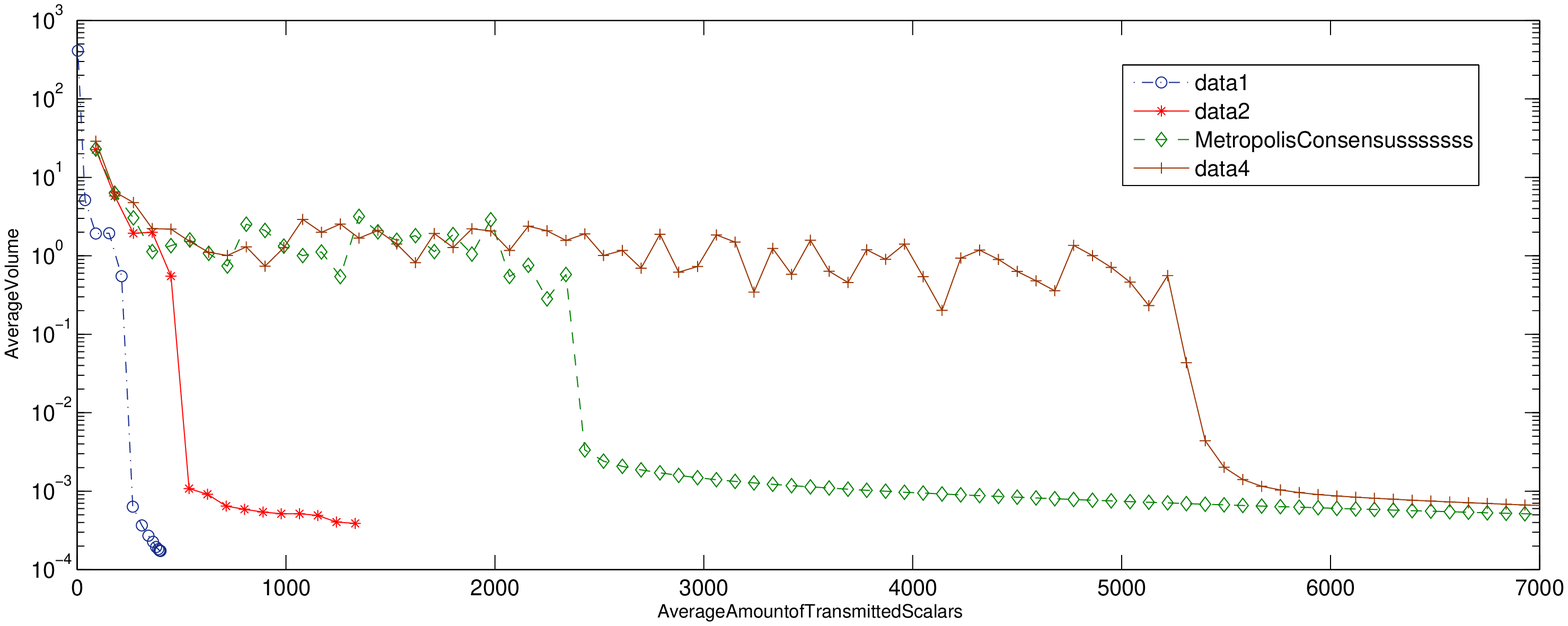} \caption{Average volume, across nodes, of the 90\% confidence region. A random unstructured network of 100 nodes is considered.}
\label{Fig:AvgVol}
\end{figure}

Similar results can be obtained on clustered networks. The number of clusters is set to $n_c=20$ and the average number of per cluster nodes is set to $\mathbb{E}[N_i^{\text{c}}]=7$. The parameter dimension is $n_p=2$, while $q=1$ and $m=10$. In particular, Fig.~\ref{Fig:AvgVolClusteredNetworks} shows the average volume, across nodes and 100 clustered network realizations, of the confidence region. Here the number of computed pairs volume-amount of data is much lower than for random trees, due to the fewer transmission rounds. The average amount of per node transmitted data, needed to obtain meaningful small volumes, is lower when employing the \ac{TAS} algorithm, as it was on random trees.

Finally, consider a random unstructured network, setting $N=100$ and $n_p=3$. As shown in Fig.~\ref{Fig:AvgVol}, in this case it is the \ac{MF} algorithm that behaves better than \ac{TAS}, providing lower volume values for the same amount of data. For comparison, it is also shown how both the \ac{MF} and the \ac{TAS} algorithm outperform the state of the art consensus algorithms, independently of the considered consensus matrix (Metropolis \cite{XiaBoy:04} or Perron \cite{OlfFaxMur:07}).

This section confirms the general behavior that was highlighted in Section~\ref{sec:ThAnalysis}: On structured topologies, such as random trees and clustered networks, there is an advantage in employing the \ac{TAS} algorithm when the network dimension is sufficiently large, and this independently of $n_p$. On unstructured networks of comparable size, the \ac{MF} produces the best results, but, in any case, the absolute amount of per node trasmitted data is much larger than in structured networks. This suggests the adoption of structured networks, together with the \ac{TAS} algorithm for the distributed computation of confidence regions, when the network traffic load for data diffusion is particularly critical.

\section{Conclusions}\label{sec:Concl}
This paper investigated the distributed evaluation of non-asymptotic confidence regions at each node in wireless sensor networks. The first main contribution is the proposal of the \ac{TAS} algorithm and its comparison with other information diffusion algorithms on structured and unstructured topologies. The second important contribution consists in demonstrating that, even in presence of truncated information diffusion, the level of confidence remains the same as in the centralized not truncated case. Simulation results provide a characterization of the trade-off for the achievable average confidence region volume as a function of the required amount of data that each node should transmit on average. The contributions nicely concur at showing that, on structured networks, the proposed \ac{TAS} algorithm is able to outperform the \ac{MF} when the network dimension is sufficiently high, this independently of the specific estimation problem dimensions, as investigated in the theoretical and numerical analyses. Future research work will be directed at combining the benefits of \ac{TAS} and \ac{MF} into a mixed approach, with the scope to further reduce the traffic burden.

\section*{Acknowledgment}
This work has been supported by the EU FP7 funded Network of Excellence Newcom\#.
\ifCLASSOPTIONcaptionsoff
  \newpage
\fi



\bibliographystyle{IEEEtran}
\bibliography{IEEEabrv,SharedBiblio}
\end{document}